\documentclass[12pt]{amsart}
\usepackage{amsmath}
\usepackage{amsxtra}
\usepackage{amstext}
\usepackage{amssymb}
\usepackage{amsthm}
\usepackage{latexsym}
\usepackage{dsfont} 
\usepackage{verbatim}
\usepackage{tabls}
\usepackage{graphicx}
\usepackage{rotating}
\usepackage{caption}
\usepackage[shortlabels]{enumitem}
\usepackage{hyperref}
\usepackage{subcaption}

\theoremstyle{plain}
\newtheorem{thm}{Theorem}[section]

\newtheorem{lem}[thm]{Lemma}
\newtheorem{conj}[thm]{Conjecture}
\newtheorem{prop}[thm]{Proposition}

\theoremstyle{definition}

\newtheorem{rem}[thm]{Remark}

\newtheorem{cla}[thm]{Claim}

\numberwithin{equation}{section}
\def\R1{\widetilde{R}}
\def\T1{\widetilde{T}}
\def\I{\mathcal{I}}

\def\supp{\operatorname{supp}}

\def\eps{\varepsilon}
\def\kap{\varkappa}
\def\R{\mathbb{R}}
\def\T{\mathcal{T}}

\def\conv{\mathcal{C}}

\def\wt{\widetilde}

\def\I{\mathcal{I}}
\def\kap{\varkappa}
\def\F{\mathcal{F}}

\def\XXint#1#2#3{{\setbox0=\hbox{$#1{#2#3}{\int}$}
     \vcenter{\hbox{$#2#3$}}\kern-.5\wd0}}

\title[R\'{e}nyi entropy]{Remarks on the R\'{e}nyi Entropy of a sum of IID random variables}
\author[Jaye]{Benjamin Jaye}
\address{School of Mathematical Sciences, Clemson University}
\email{bjaye@clemson.edu}
\author[Livshyts]{Galyna V. Livshyts}
\address{Department of Mathematics, Georgia Tech}
\email{glivshyts6@math.gatech.edu}
\author[Paouris]{Grigoris Paouris}
\address{Department of Mathematics, Texas A\&M}
\email{grigoris@math.tamu.edu}
\author[Pivovarov]{Peter Pivovarov}
\address{Department of Mathematics, University of Missouri}
\email{pivovarovp@missouri.edu}

\begin{document}
\maketitle

\begin{abstract} 
In this note we study a conjecture of Madiman and Wang \cite{MW} which
predicted that the generalized Gaussian distribution minimizes the
R\'{e}nyi entropy of the sum of independent random variables. Through
a variational analysis, we show that the generalized Gaussian fails to
be a minimizer for the problem.
\end{abstract}

\section{Introduction}

For $p>1$, the $p$-R\'{e}nyi \cite{Re} entropy of a (continuous)
random vector $X$ in $\R^d$ distributed with density $f$ is defined by
$$h_p(X) = - \frac{1}{p-1}\log \int_{\R^d} f(x)^p d\mu_d(x) = -\frac{1}{p-1}\log \|f\|_p^p,
$$ where $\mu_d$ denotes the $d$-dimensional Lebesgue measure.  As
$p\to 1^+$, $h_p(X)$ converges to the usual Shannon
entropy \begin{equation*} h(X) = -\int_{\R^d} f(x)\log f(x)
  d\mu_d(x)\end{equation*} (provided that the density of $X$ is
sufficiently regular to justify passage of the limit). For the entropy
power $N(X)=\exp(2h(X)/d)$, the fundamental entropy power inequality
(EPI) of Shannon \cite{Shan} asserts that for independent random
vectors $X_1$ and $X_2$,
\begin{equation*}
  N(X_1+X_2)\geq N(Z_1+Z_2),
\end{equation*}where $Z_1$, $Z_2$ are independent Gaussians satisfying
$N(X_i)= N(Z_i), i=1,2$. A firm connection between the EPI,
$p$-R\'{e}nyi entropy and fundamental results like the Brunn-Minkowski
and Young's convolution inequalities goes back to Dembo, Cover and
Thomas \cite{DCT}.  See Principe \cite{Pr} for more information about
where the R\'{e}nyi entropy arises; see also Bobkov, Marsiglietti
\cite{BM} for a related discussion.

Recently, there has been increasing interest in $p$-R\'{e}nyi entropy
inequalities. Interestingly, the following basic mathematical question
is still open: \emph{Over all random variables $X$ with $h_p(X)$ some
  fixed quantity, what are the minimizers of the entropy $h_p(X+X')$,
  where $X'$ is an independent copy of $X$?}  We learnt about this
question from the papers of Madiman, Melbourne, Xu, and Wang \cite{MW,
  MMX}, who studied unifying entropy power inequalities for the
R\'{e}nyi entropy, which, in the limit $p\to 1^+$ recover the
statement that, over all probability distributions with $h(X)$ fixed,
$h(X+X')$ is minimized if (and only if) $X$ is a Gaussian, see
e.g. \cite{DCT}.

Several closely related questions have been recently addressed
involving the $p$-R\'enyi entropy power $N_p(X) =
\exp(\tfrac{2}{d}h_p(X))$.  Bobkov and Chistyakov \cite{3} show that
there is a constant $c>0$, depending on $d$ and $p$, such that
$N_p(\sum_{j=1}^n X_j)\geq c\sum_{j=1}^n N_p(X_j)$ for independent
random vectors $X_1, \dots, X_n$.  A sharper form of the constant was
subsequently found by Ram and Sason \cite{8}.  Bobkov and Marsiglietti
\cite{1} proved that $N_p(X_1+X_2)^{\alpha}\geq N_p(X_1)^{\alpha}
+N_p(X_2)^{\alpha}$ for $X_1,X_2$ independent Random vectors if
$\alpha\geq \frac{p+1}{2}$. There has been considerable further recent
success extending the EPI to the R\'enyi setting \cite{2, 4, 5, 7, 8,
  10}.


Following \cite{LYZ, MW, MMX}, for $\beta>0$, consider the
\emph{Generalized Gaussian}
$$G_{\beta,p}(x) = \alpha (1-\beta|x|^2)_+^{1/(p-1)},
$$ where $\alpha$ is chosen so that $\int_{\R^d} G_{\beta,p}(x)
d\mu_d(x) = 1$. The generalized Gaussian is the distribution with the
smallest second moment with a given R\'{e}nyi entropy, see work of
Lutwak, Yang, and Zhang \cite{LYZ}, as well as earlier results of
Costa, Hero, and Vignat \cite{CHV}.  Madiman and Wang made the
following bold conjecture (Conjecture IV.3 in \cite{MW}).

\begin{conj}[The Madiman-Wang Conjecture] If $X_j$,
$j=1,\dots,n$, are independent random variables with densities $f_j$,
  and $Z_j$ are independent random variables distributed with respect
  to $G_{\beta_j,p}$ where $\beta_j$ is chosen so that $h_p(X_j) =
  h_p(Z_j)$, then $$h_p(X_1+\dots+X_n)\geq
  h_p(Z_1+\dots+Z_n).$$\end{conj}

This conjecture has been confirmed in the case $p=+\infty$, see \cite{6, 12}.

In this note we will show that unfortunately this conjecture does not
hold in the special case when $d=1$, $p=2$, $n=2$ and $X_1$ and $X_2$
are identically distributed, see Section \ref{MWsec}.  However, we do
suspect that a minimizing distribution is a relatively small
perturbation of the generalized Gaussian.

Throughout this note we only consider the case where $X_1,\dots ,X_n$ are independent
copies of a random variable $X$ with density $f$.  The question of
finding the minimizer of $h_p(X_1+\dots+ X_n)$ with $h_p(X)$ fixed can
then be rephrased as a constrained maximization problem, which we introduce in Section \ref{constrained}.  Subsequently, in Section \ref{variation} we take the first variation of this maximization problem.  We have not been able to develop a satisfactory theory of the associated Euler-Lagrange equation (\ref{neccond}), but we show in Section \ref{MWsec} that the generalized Gaussian is not a solution to (\ref{neccond}), and so fails to be a maximizer of the extremal problem.  We conclude the paper with some elementary remarks and speculation.


\medskip
\medskip
\textbf{Acknowledgement.} The first named author is supported by NSF DMS-1830128, DMS-1800015 and NSF CAREER DMS-1847301. The second named author is supported by the NSF CAREER DMS-1753260. The third named author is supported by the NSF DMS-1812240. The fourth named author is supported by the NSF DMS-1612936. The work was partially supported by the National Science Foundation under Grant No. DMS-1440140 while the authors were in residence at the Mathematical Sciences Research Institute in Berkeley, California, during the Fall 2017 semester.

The authors are especially grateful to the reviewers for valuable comments and suggestions, which helped improve the paper and clarify the exposition.

\section{The constrained maximization problem}\label{constrained}

Denote by $\conv_n(f)$ the $(n-1)$-fold convolution of a given
function $f$ with itself, that is, $\conv_n(f) = f*f*\cdots *f$, where
there are $n$ factors of $f$ (and $n-1$ convolutions).  Then
$\conv_1(f)=f$.  It will be convenient to set $\conv_0(f)=
\delta_0$, the Dirac delta measure, so that $g*\conv_0(f)=g$ for any measurable function $g$.

Throughout the text, we fix $M>0$, $n\in \mathbb{N}$ and $p\in
(1,\infty)$. We set
$$\F = \bigl\{f \in L^1(\R^d)\cap L^p(\R^d), \, f\geq 0,\,
\|f\|_p^p=M,\, \|f\|_1=1\bigl\}
$$ and consider the extremal problem
\begin{equation}\label{convprob}\begin{cases}\;\text{Maximize } \I(f)\stackrel{\operatorname{def}}{=}\int_{\R^d}[\conv_n(f)(x)]^pd\mu_d(x)\\\text{ subject to } f\in \F.
\end{cases} \end{equation} Put \begin{equation}\label{lambdadef}\Lambda = \Lambda(p,M) =
  \sup\{\I(f): f\in \F\}.\end{equation}

We begin with a simple scaling lemma, which we will use often in what follows.
\begin{lem}\label{scaling}  Suppose that $f\in L^1(\mathbb{R}^d)\cap L^p(\mathbb{R}^d)$ is non-negative, and $\|f\|_1>0$.  The function
$$\wt f = \frac{1}{\lambda^d
    \|f\|_1}f\Bigl(\frac{\cdot}{\lambda}\Bigl), \text{ with }\lambda =
  \Bigl(\frac{\|f\|_p^p}{M\|f\|_1^p}\Bigl)^{\tfrac{1}{d(p-1)}},$$
  belongs to $\F$, and
$$\I(\wt f) = \frac{M}{\|f\|_p^p}\frac{1}{\|f\|_1^{p(n-1)}}\I(f).
$$
\end{lem}

\begin{proof}
Observe that, for any $r\in [1,\infty)$,
$$\|\wt f\|_r^r = \frac{1}{\lambda^{d(r-1)}\|f\|_1^r}\|f\|_r^r.
$$
Plugging in $r=1$ and $r=p$ (and recalling the definition of $\lambda$) we see that $\wt f\in \F$.  Next, observe that
$$\conv_n(\wt f)(x) = \frac{1}{\|f\|_1^n\lambda^d}\conv_n(f)\bigl(\frac{x}{\lambda}\bigl)\text{ for any }x\in \R^d.
$$
Whence, $$\I(\wt{f}) = \frac{1}{\lambda^{d(p-1)}\|f\|_1^{pn}}\I(f),$$
and the proof is complete by recalling the definition of $\lambda$.
\end{proof}

We next prove that (\ref{convprob}) has a maximizer.  A radial function
$f$ on $\R^d$ is called decreasing if $f(y)\leq f(x)$ whenever
$|y|\geq |x|$.

\begin{prop}\label{existence}
The problem (\ref{convprob}) has a lower-semicontinuous, radially
decreasing, maximizer $Q$.
\end{prop}

\begin{proof} First observe that for any measurable function $f$, iterating Riesz's rearrangement inequality \cite[Theorem
  3.7]{LL} yields $\I(f)\leq \I(f^*)$, where $f^*$ is the
  symmetric rearrangement of $f$; see  \cite[Section 3.4]{B} for
  related multiple convolution rearrangement inequalities and their
  equality cases.  Also, notice that if $f\in \F$, then $f^*\in \F$.

Take non-negative functions $f_j\in \F$ such that $\Lambda
=\lim_{j\to\infty}\I(f_j)$ (recall $\Lambda$ from (\ref{lambdadef})).
By replacing $f_j$ with its symmetric rearrangement, we may assume
that $f_j$ are radial and decreasing.  Passing to a subsequence if
necessary, we may in addition assume that $f_j\to f$ weakly in
$L^p(\R^d)$.  Consequently, $f$ is radial, decreasing, $f\geq 0$, and
$\|f\|_p^p\leq M$.  (To see this, observe that the set of radial
decreasing nonnegative functions with norm at most $M^{1/p}$ is a
closed convex set in $L^p(\R^d)$, so by Mazur's Lemma, see
e.g. \cite[Theorem 2.13]{LL}, this set is weakly closed.)  By
modifying $f$ on a set of measure zero if necessary, we may assume
that $f$ is lower semi-continuous\footnote{If $f$ is discontinuous at
  $x\in \R^d$, then define $f(x) = \sup_{|y|>|x|}f(y)$ (i.e. the
  one-sided radial limit from the right). Then $\{f>\lambda\}$ is open
  for every $\lambda>0$.}.





%

\begin{cla}\label{claim:pointwise} As $j\to \infty$, $f_j\to f$ $\mu_d$-almost everywhere.\end{cla}

\begin{proof} For $r>0$, define $v_j(r) = f_j(x)$ and $v(r)=f(x)$ whenever $|x|=r$.  Then since
$f_j$ converges weakly to $f$ in $L^p(\R^d)$, we have that whenever
 $I$ is a closed interval of finite Lebesgue measure in $(0,\infty)$,
$$\lim_{j\to \infty}\int_I v_j(s) d\mu_1(s) = \int_I v(s) d\mu_1(s).
$$
Insofar as the function $v$ is non-decreasing, it has at most countably many points of discontinuity.
If $r>0$ is a point of continuity of $v$, and $I_{k} = [r-2^{-k}, r]$,
then
$$v(r) = \lim_{k\to\infty}\frac{1}{2^{-k}}\int_{I_{k}}v(s) d\mu_1(s) =
\lim_{k\to\infty}\lim_{j\to\infty}\frac{1}{2^{-k}}\int_{I_{k}}v_j(s) d\mu_1(s)
$$ but since $v_j$ is decreasing we have that $v_j(s)\geq v_j(r)$ for
$s\in I_{k}$.  Thus
$$v(r) \geq \limsup_{j\to\infty}v_j(r).
$$
Arguing similarly with intervals whose left end-point is $r$, we also have that
$$v(r) \leq \liminf_{j\to\infty}v_j(r).
$$ Thus $\lim_{j\to \infty}v_j = v$ at every point of continuity of $v$.  If $E$ is a countable set in $(0,\infty)$, then $E\times \mathbb{S}^{d-1}$ is a Lebesgue null set in $\R^d$, so the claim follows.
\end{proof}

Notice that, as a consequence of this claim, Fatou's Lemma ensures that $\|f\|_1\leq 1$.  Our next claim is

\begin{cla}\label{claim:lq}  If $1<q<p$, then $f_j \to f$ strongly in $L^q(\R^d)$ as $j\to \infty$.
\end{cla}

The proof of this claim is a variant of the Vitali convergence theorem (see e.g. Theorem 9.1.6 of \cite{Ros}), but observe that it does not necessarily hold if one was to remove the radially decreasing property of the functions $f_j$ (just consider a sequence of translates of a fixed function).

\begin{proof}  Fix $\eps>0, \delta>0$.  Insofar as the functions $f_j$ and $f$ are radially decreasing,
$$\bigcup_j\{|f_j|\geq \tfrac{\delta}{2}\}\cup\{|f|\geq\tfrac{\delta}{2}\}\subset B,
$$ where $B$ is the closed ball centered at $0$ of radius
  $\bigl(\frac{2}{\mu_d(B(0,1))\delta}\bigl)^{1/d}.$ (Otherwise we
  would have $\|f_j\|_1>1$ for some $j$, or $\|f\|_1>1$.)

On $\R^d\backslash B$, we have $|f_j|<\delta/2$ for every $j$, and $|f|<\delta/2$, whence
$$\int_{\R^d\backslash B}|f_j(x)-f(x)|^qd\mu_d(x)\leq \delta^{q-1}\Bigl(\|f_j\|_1+\|f\|_1\Bigl)\leq 2\delta^{q-1}<\frac{\eps}{3}
$$
provided $\delta>0$ is chosen sufficiently small.

Now fix $\kap>0$.  Observe that,
$$\int_{B\cap\{|f_j-f|<\kap\}}|f_j(x)-f(x)|^qd\mu_d(x) \leq \mu_d(B)\kap^q<\frac{\eps}{3}
$$
if $\kap$ is chosen sufficiently small.  On the other hand, since $B$ has finite measure, one can invoke continuity of measure from above, thus we have that $f_j\to f$ in measure on $B$ as $j\to \infty$.  From the inequalities
\begin{equation}\begin{split}\nonumber\int_{B\cap\{|f_j-f|\geq \kap\}}|f_j(x)-f(x)|^qd\mu_d(x)
&\leq \mu_d(B\cap \{|f_j-f|\geq \kap\})^{1-q/p}
\|f_j-f\|_p^q\\
&\leq 2^qM^{q/p} \mu_d(B\cap \{|f_j-f|\geq \kap\})^{1-q/p},
\end{split}\end{equation}
we infer that there exists $N\in \mathbb{N}$ such that
$$\int_{B\cap\{|f_j-f|\geq \kap\}}|f_j(x)-f(x)|^qd\mu_d(x)<\frac{\eps}{3} \text{ for all }j\geq N.
$$
Bringing these estimates together, it follows that $\|f_j-f\|_q^q<\eps$ for every $j\geq N$.
\end{proof}

Our next goal is to use this claim in order to show that $\I(f)=\Lambda$.  To this end, observe that repeated application of Young's convolution inequality \cite{LL}
  yields that, for any $n$-tuple of functions $g_1,\dots, g_n$,
\begin{equation}\label{nyoung}\Bigl(\int_{\R^d}|g_1*g_2*\cdots *g_n(x)|^p d\mu_d(x) \Bigl)^{1/p}\leq \prod_{j=1}^n\|g_j\|_{(np')'},
\end{equation}
where $p' = p/(p-1)$ is the H\"{o}lder conjugate of $p$, so $(np')' = \tfrac{np}{np-p+1}$.  Since $n>1$, $(np')'\in (1,p)$.

To apply this inequality, first use Minkowski's inequality to observe that, $$|\I(g_1)^{1/p}-\I(g_2)^{1/p}|\leq \Bigl(\int_{\R^d} |\conv_n(g_1)(x) - \conv_n(g_2)(x)|^p d\mu_d(x)\Bigl)^{1/p},$$ but,
$$\conv_n(g_1)-\conv_n(g_2) = \sum_{k=0}^{n-1}\conv_k(g_1)*(g_1-g_2)*\conv_{n-k-1}(g_2),
$$
and hence
$$|\I(g_1)^{1/p}-\I(g_2)^{1/p}|\leq\sum_{k=0}^{n-1}\Bigl(\int_{\R^d} |\conv_k(g_1)*(g_1-g_2)*\conv_{n-k-1}(g_2)(x)|^pd\mu_d(x)\Bigl)^{1/p}.
$$
Appealing to (\ref{nyoung}) now yields,
$$|\I(g_1)^{1/p}-\I(g_2)^{1/p}|\leq\sum_{k=0}^{n-1}\|g_1\|_{(np')'}^k\|g_2\|_{(np')'}^{n-k-1}\|g_1-g_2\|_{(np')'}.
$$ Returning to our sequence $f_j$, it is a consequence of
H\"{o}lder's inequality that $\|f_j\|_{(np')'}\leq
\|f_j\|_1^{\theta}\|f_j\|_p^{1-\theta}$ with some $\theta\in (0,1)$
depending on $n$ and $p$, so $\|f_j\|_{(np')'}\leq C(M,n,p)$ (and the
same inequality holds with $f_j$ replaced by $f$).  Whence there is a
constant $C'(n,p,M)$ such that
$$|\I(f_j)^{1/p}-\I(f)^{1/p}| \leq C'(n,p,M)\|f_j-f\|_{(np')'} \,\text{ for every $j$}.
$$

Since $(np')'\in (1,p)$, Claim \ref{claim:lq} yields that
$f_j \to f$ in $L^{(np')'}$ as $j\to \infty$. Hence $\I(f)=\lim_{j\to \infty}\I(f_j)=\Lambda$. (It follows that $f$
is not identically zero.)

It remains to show that $f\in \F$.  To this
end, we apply Lemma \ref{scaling}: Consider the function
$$\wt{f} = \frac{1}{\|f\|_1\lambda^d}f\Bigl(\frac{\cdot}{\lambda}\Bigl), \text{ with }\lambda = \Bigl(\frac{\|f\|_p^p}{M\|f\|_1^p}\Bigl)^{\frac{1}{d(p-1)}}.$$  Then $\wt f\in \F$ and $\I(\wt f ) = \frac{M}{\|f\|_p^p}\frac{1}{\|f\|_1^{p(n-1)}}\Lambda.$
Consequently, if $\|f\|_p^p<M$ or $\|f\|_1<1$, then $\I(\wt f)>\Lambda$, which is absurd.  Thus $f\in \F$ and the proof of the proposition is complete.
\end{proof}

\section{The First Variation}\label{variation}

With the existence of a maximizer proved, we now wish to analyze it
analytically.

To introduce the Euler-Lagrange equation associated to (\ref{convprob}) it will be convenient to define, for a function $f$, $\T(f)(x) = f(-x)$.  Observe that, if $f,g,h$ are non-negative measurable functions,
\begin{equation}\label{threeconv}\int_{\R^d} f(x)(\T(g)*h)(x)d\mu_d(x) = \int_{\R^d}(f*g)(x)h(x) d\mu_d(x).
\end{equation}

\begin{prop}  A lower-semicontinuous function $Q\in \F$ is a maximizer of the problem (\ref{convprob}) if and only if
\begin{equation}\label{neccond}
[\T(\conv_{n-1}(Q))]*[\conv_n(Q)]^{p-1} = \frac{\Lambda}{M n}
Q^{p-1}+\frac{\Lambda (n-1)}{n} \text{ on }\{Q>0\}.
\end{equation}
\end{prop}

\begin{rem}\label{raddecrem}Observe that if $Q$ is radially decreasing, then $\conv_{n-1}(Q)$ is again radially decreasing for any $n\in \mathbb{N}$, so $\T(\conv_{n-1}(Q)) = \conv_{n-1}(Q)$ in this case.\end{rem}

\begin{proof} The sufficiency is easy to show.  Integrating both sides of (\ref{neccond}) against $Q$, and recalling that $Q\in \mathcal{F}$, yields
$$\int_{\R^d} Q(x)\cdot ([\T(\conv_{n-1}(Q))]*[\conv_n(Q)](x))^{p-1}d\mu_d(x) =\Lambda.
$$ But using Tonelli's theorem and (\ref{threeconv}), the left hand side is equal to
  $\int_{\R^d}(\conv_n(Q)(x))^pd\mu_d(x) = \I(Q)$.  

Conversely, consider a bounded function $\varphi$ compactly supported
in the open set $\{Q>0\}$.  Since $Q$ is lower-semicontinuous,
$\inf_{\supp(\varphi)} Q>0$.  Therefore, (insofar as $\varphi$ is bounded) there exists a constant $C>0$ such that
\begin{equation}\label{phismall}
|\varphi|\leq C Q \text{ on }\R^d,
\end{equation}
 so in particular, there exists $t_0>0$ such that for $|t|\leq t_0$ it
 follows that $Q_t\stackrel{\operatorname{def}}{=}Q+t\varphi$ is
 non-negative.  In the notation of Lemma \ref{scaling} with $f=Q_t$,
 we consider the function
$${\wt Q_t} = \frac{1}{\lambda^d}\frac{(Q + t
   \varphi)\bigl(\frac{\cdot}{\lambda}\bigl)}{\|Q+t\varphi\|_1},$$
 with the corresponding $\lambda>0$ satisfying $\|\wt Q_t\|^p_{p} =
 \|Q\|_{p}^p=M$.
 Of course we also have $\int_{\R}{\wt Q_t}(x)\,d\mu_d(x)=1$ regardless of
 $\lambda$ for $|t|<t_0$.  We conclude that ${\wt Q_t}$ belongs to
 $\mathcal{F}$, and therefore \begin{equation}\label{Qtsmaller}\I(\wt Q_t)\leq \I(Q)=\Lambda,\text{ for
 all }|t|<t_0.\end{equation} Moreover, as in Lemma \ref{scaling},
\begin{equation}\label{Qtcalc}
\I(\wt Q_t) = \frac{1}{\lambda^{d(p-1)}\|Q+t\varphi\|_1^{np}}\int_{\R^d}[\conv_n(Q+t\varphi)(x)]^pd\mu_d(x).
\end{equation}


For $|t|<t_0$, we calculate, using commutativity and associativity of the convolution operator,
$$\frac{d}{dt} \conv_n(Q+t\varphi)^p = pn[\varphi*\conv_{n-1}(Q+t\varphi)][\conv_n(Q+t\varphi)]^{p-1},
$$
and
\begin{equation}\begin{split}\label{2ndder}
\frac{d^2}{dt^2}\conv_n(Q+&t\varphi)^p =  pn(n-1)    \varphi*\varphi*\conv_{n-2}(Q+t\varphi)[\conv_n(Q+t\varphi)]^{p-1}\\&+n^2p(p-1)[\varphi*\conv_{n-1}(Q+t\varphi)]^2[\conv_n(Q+t\varphi)]^{p-2}.
\end{split}
\end{equation}
Crudely employing the bound (\ref{phismall}) in (\ref{2ndder}), we infer that there is a constant $C>0$, depending on $n$, $p$ and $t_0,$ such that for all $|t|<t_0$,
$$\Bigl|\frac{d^2}{dt^2}\conv_n(Q+t\varphi)^p\Bigl|\leq C\conv_n(Q)^p.
$$
Whence, the second order Taylor formula yields that
\begin{equation}\begin{split}\label{pointwiseperturb}|\conv_n(Q+&t\varphi)^p - \conv_n(Q)^p - npt [\varphi*\conv_{n-1}(Q)][\conv_n(Q)]^{p-1}| \leq Ct^2\conv_n(Q)^p,\end{split}\end{equation}
for $|t|< t_0$.  Integrating the pointwise inequality (\ref{pointwiseperturb}) yields
\begin{equation}\begin{split}\label{b4lambdaperturb}
\int_{\R^d}\!&[\conv_n(Q+t\varphi)(x)]^pd\mu_d(x)\\&= \!\Lambda\!+\!\!npt\int_{\R^d}[\varphi*\conv_{n-1}(Q)(x)][\conv_n(Q)(x)]^{p-1} d\mu_d(x)+
      O(t^2)
\end{split}\end{equation}
as $t\to 0$.

Now, recalling the definition of $\lambda$, we calculate
\begin{equation}\begin{split}\label{scaleexpansion}&\lambda^{d(p-1)}\|Q+t\varphi\|_1^{np} = \frac{\|Q+t\varphi\|_p^p}{M}\|Q+t\varphi\|_1^{(n-1)p}\\
& = \Bigl(1+\frac{pt}{M}\int_{\R^d}\varphi(x) Q(x)^{p-1}
    d\mu_d(x)+O(t^2)\Bigl)\\
    &\;\;\;\cdot\Bigl(1+t(n-1)p\int \varphi(x)
    d\mu_d(x)+O(t^2)\Bigl),\end{split}\end{equation}
where in the expansion of $\|Q+t\varphi\|_p^p$ we have again used the inequality (\ref{phismall}) to obtain the $O(t^2)$ term.

Plugging the two expansions (\ref{scaleexpansion}) and (\ref{b4lambdaperturb}) into (\ref{Qtcalc}) yields that, as $t\to 0$,
\begin{equation}\begin{split}\nonumber\I(\wt Q_t) = &\Lambda + pt\Bigl\{n\int_{\R^d}[\varphi*\conv_{n-1}(Q)(x)][\conv_n(Q)(x)]^{p-1} d\mu_d(x)\\
& - \frac{\Lambda}{M}\int_{\R^d}\varphi(x) Q^{p-1}(x)d\mu_d(x) - (n-1)\Lambda\int_{\R^d}
    \varphi(x) \,d\mu_d(x)\Bigl\} +O(t^2).
\end{split}\end{equation}

From (\ref{Qtsmaller}) it follows that $\lim_{t\to 0}\frac{\I(\wt Q_t)-\I(Q)}{t}=0$, so the second term in the prior expansion must vanish, that is,
$$\int_{\R^d}\varphi(x) \Bigl\{n[\T(\conv_{n-1}(Q))]*[\conv_n(Q)]^{p-1}(x) -
\frac{\Lambda}{M} Q^{p-1} - (n-1)\Lambda\Bigl\}d\mu_d(x)=0,
$$
where (\ref{threeconv}) has been used.  Since $\varphi$ was any bounded function compactly supported in
$\{Q>0\}$, we conclude that (\ref{neccond}) holds.\end{proof}

\section{On the Madiman-Wang conjecture}\label{MWsec}

\begin{prop}\label{keyprop}
The generalized Gaussian is not the extremizer for problem (\ref{convprob}).
\end{prop}
\begin{proof} Consider the simplest case $d=1$, $p=2$, and $n=2$.  We shall show
that the function $G(x) = \alpha(1-|x|^2)_+$ does not satisfy the
equation
\begin{equation}\label{3conv}\mathcal{C}_3(f)=af+b\text{ on
  }[-1,1]\text{ with }a,b>0,
  \end{equation}
  and so no function of the
form $\frac{c}{\lambda}G(\frac{\cdot}{\lambda})$, with $c,\lambda>0$,
satisfies (\ref{neccond}), for any value of $\Lambda$ (recall Remark \ref{raddecrem}).  In fact,
we shall show that $\mathcal{C}_3(G)=G*G*G$ is not a quadratic
polynomial near $0$.

For this, observe:
$$G'' = 2\alpha(\delta_{-1} - \chi_{[-1,1]} +\delta_{1}).
$$ Thus, $(G*G*G)'''''' = (G''*G''*G'')$ is the threefold convolution
of the above measure.  The threefold convolution of $-2\chi_{[-1,1]}$
equals $-8(3-|x|^2)_+$ on $[-1,1]$, and no other term in the convolution
$G''*G''*G''$ is quadratic in $|x|$.  Therefore, $G*G*G$ has
non-vanishing sixth derivative at $0$, but $a+bG$ does have
vanishing sixth derivative at $0$.
\end{proof}

\begin{rem} Moreover, for any dimension $d$, the random vector $X$ in $\R^d$ with i.i.d. coordinates $X_i$, each distributed according to the generalized Gaussian density, does not constitute the extremizer for this problem. Indeed, in this case $h_p(X)=dh_p(X_i)$, and it remains to use Proposition \ref{keyprop}. Therefore, a random vector with i.i.d. coordinates which are generalized Gaussians is not an extremal case for this question.
\end{rem}

\section{Any radially decreasing solution of (\ref{neccond}) is compactly supported}

In this section, we discuss the following

\begin{prop}
Decreasing radial solutions of (\ref{neccond})
are compactly supported.
\end{prop}
\begin{proof} Suppose that $Q\in \F$ solves (\ref{neccond}) and $Q$
is not compactly supported.  Since $Q$ is non-negative and radially
decreasing, its support is $\R^d$.

The term $G = \conv_{n-1}(Q)*(\conv_n(Q))^{p-1}$ on the left hand side
of (\ref{neccond}) belongs to $L^r$, where $r=\max(1, 1/(p-1))$.
Indeed, if $p\geq 2$ then $\int_{\R^d} G(x) d\mu_d(x) =
\int_{\R^d}[\conv_n(Q)(x)]^{p-1}d\mu_d(x)$ (recall that $Q\geq 0$ with
$\int_{\R^d}Q(x)d\mu_d(x)=1$), but
$\int_{\R^d}\conv_n(Q)(x)d\mu_d(x)=1$ and
$$\int_{\R^d}\conv_n(Q)^p(x)d\mu_d(x)=\Lambda <\infty,$$ so
$\conv_n(Q)\in L^{p-1}(\R^d)$.  If $1<p<2$, then $t\mapsto
t^{1/(p-1)}$ is convex, so by Jensen's inequality, $G^{1/(p-1)}\leq
\conv_{n-1}(Q)*(\conv_n(Q)^{p-1})^{1/(p-1)} = \conv_{2n-1}(Q)$, whence
$\|G\|_{1/(p-1)}\leq 1$ in this case.

On the other hand, the right hand side of (\ref{neccond}) belongs to $L^r$ only if $\Lambda=0$, which is absurd, since $\mathcal{F}$ certainly contains non-zero functions. 
\end{proof}



\section{Remarks}

In this section we make some remarks that suggest that although
the generalized Gaussian is not an optimal distribution for the
problem (\ref{convprob}), a reasonably small perturbation of the
generalized Gaussian could well be.

Beginning with $f_0(x) = \mathbf{1}_{[-1,1]}$, consider the following iteration for $j\geq 1$
$$f_{j}(x) = \frac{\mathcal{C}_3(f_{j-1})(x) -
  \mathcal{C}_3(f_{j-1})(1)}{\mathcal{C}_3(f_{j-1})(0) -
  \mathcal{C}_3(f_{j-1})(1)}.
$$

Numerically, this iteration converges pointwise to a
solution of the equation (\ref{3conv}) for some $a,b>0$ satisfying the
constraints $f(0)=1$ and $f(1)=0$ (so the support of $f$ is $[-1,1]$).
The resulting function $f$ can then be re-scaled via the
transformation $\frac{c}{\lambda}f(\tfrac{\cdot}{\lambda})$
($c,\lambda>0$) to have any given positive integral and $L^2$-norm.
We do not know if the solution of $\mathcal{C}_3(f) = af+b$ is unique
(modulo natural invariants in the problem), so we cannot say that this
function $f$ corresponds to a solution of the constrained maximization
problem (\ref{convprob}).

We provide the graphs of $f_1, f_2, f_3$ and $f_4$ (see Figure
\ref{fig:iterates} below), and the algebraic expressions for $f_1$,
$f_2$ and $f_3$ on $[-1,1]$.
\begin{gather*} f_1(x) = 1 - x^2,\; f_2(x) = 1 - \frac{6 x^2}{5} + \frac{x^4}{5}\\
f_3(x) = 1 - \frac{62325 x^2}{50521} + \frac{12810 x^4}{50521} -
\frac{1050 x^6}{50521} + \frac{45 x^8}{50521} -
\frac{x^{10}}{50521}.\end{gather*}

\begin{figure}[h!]
  \centering
  \begin{subfigure}[b]{0.2\linewidth}
    \includegraphics[width=\linewidth]{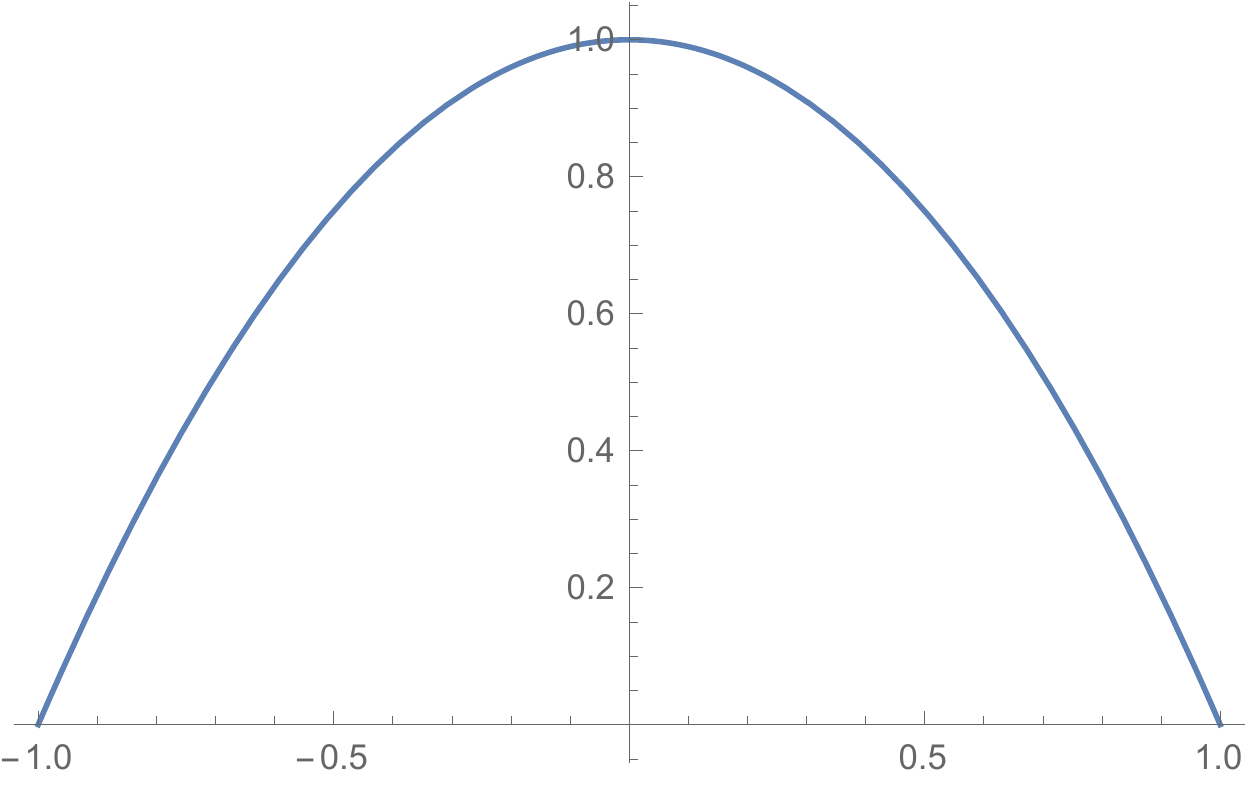}
     \caption{$f_1$.}
  \end{subfigure}
  \begin{subfigure}[b]{0.2\linewidth}
    \includegraphics[width=\linewidth]{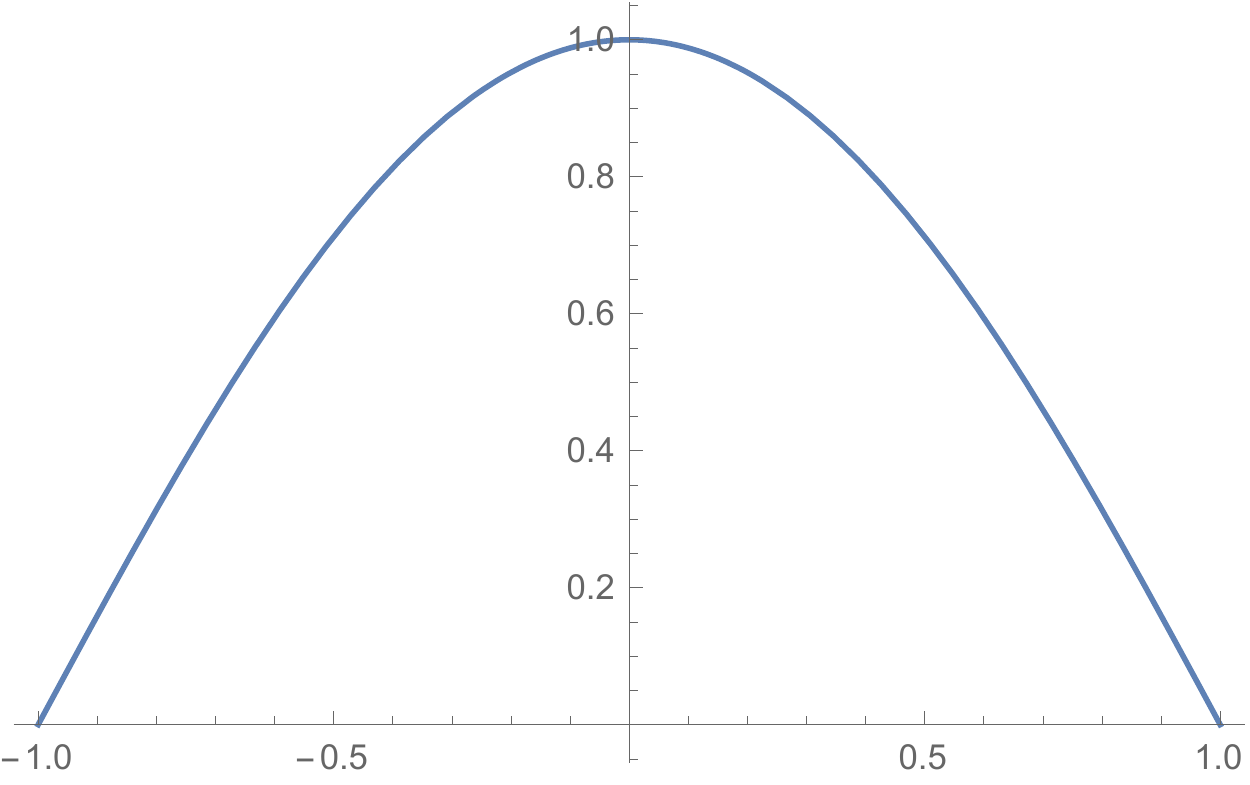}
    \caption{$f_2$.}
  \end{subfigure}
  \begin{subfigure}[b]{0.2\linewidth}
    \includegraphics[width=\linewidth]{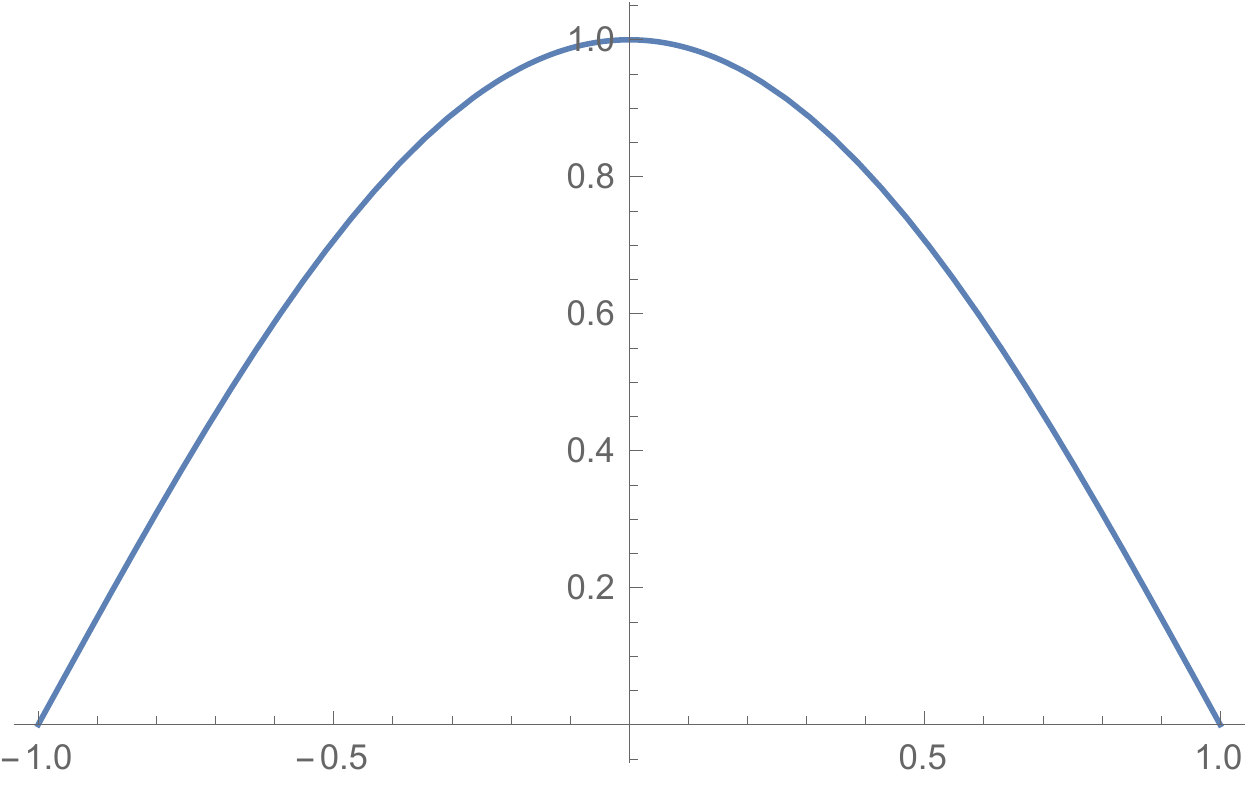}
    \caption{$f_3$.}
  \end{subfigure}
  \begin{subfigure}[b]{0.5\linewidth}
    \includegraphics[width=\linewidth]{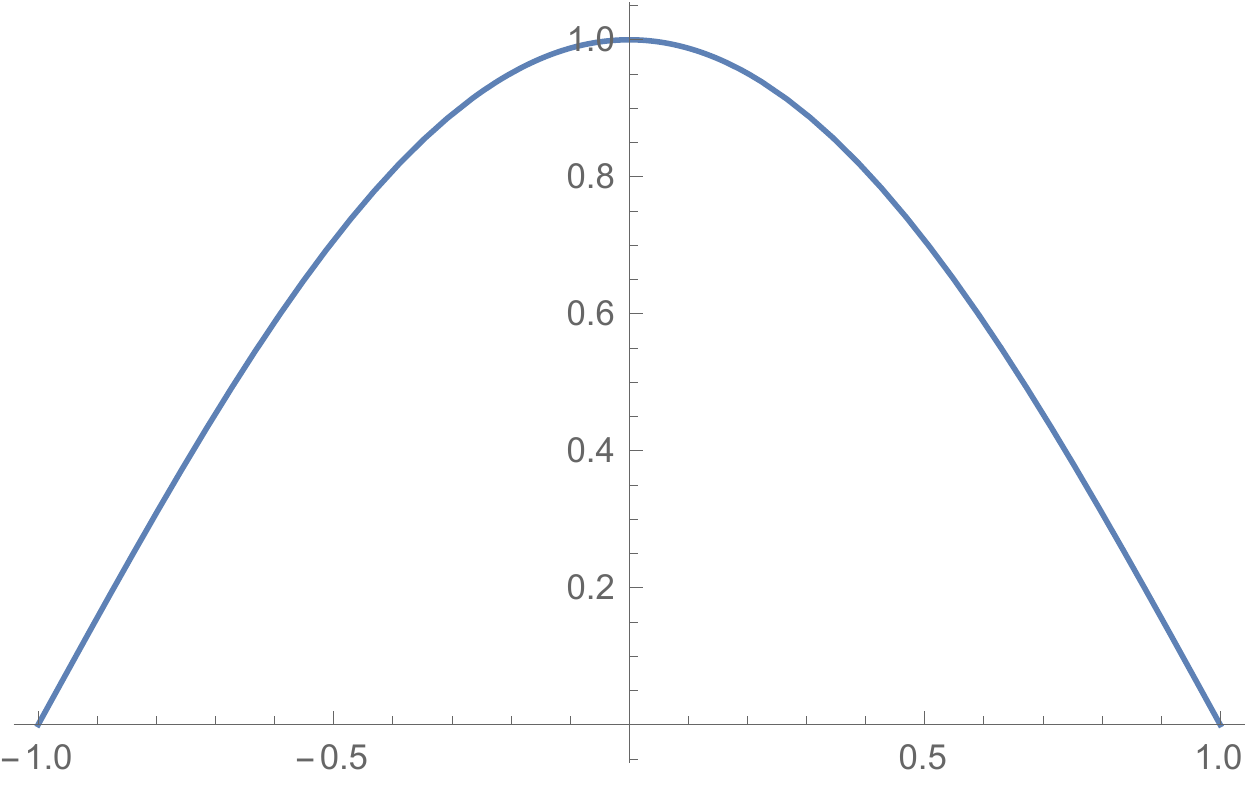}
    \caption{$f_4$.}
  \end{subfigure}
  \caption{The graphs of $f_1, \dots, f_4$ on $[-1,1]$.}
  \label{fig:iterates}
\end{figure}


\pagebreak

\end{document}